\documentclass[journal]{IEEEtran}
\usepackage{amsmath,amsfonts,mathrsfs,graphicx,xcolor,citesort}

\newtheorem{theorem}{Theorem}

\begin{document}
\title{On the Sum Rate of a 2 $\times$ 2 Interference Network}

\author{Murali Sridhar,	\and 	Srikrishna Bhashyam \thanks{This work was done at the Department of Electrical Engineering, IIT Madras, Chennai, India. Murali Sridhar is with Intel Technology India Pvt. Ltd., Bangalore. Srikrishna Bhashyam is with the Department of Electrical Engineering, IIT Madras.}}
\maketitle

\begin{abstract}
In an $M \times N$ interference network, there are $M$ transmitters and $N$ receivers with each transmitter having independent messages for each of the $2^{N}-1$ possible non-empty subsets of the receivers. We consider the 2 $\times$ 2 interference network with 6 possible messages, of which the 2 $\times$ 2 interference channel and $X$ channel are special cases obtained by using only 2 and 4 messages respectively. Starting from an achievable rate region similar to the Han-Kobayashi region, we obtain an achievable sum rate. For the Gaussian interference network, we determine which of the 6 messages are sufficient for maximizing the sum rate within this rate region for the low, mixed, and strong interference conditions. It is observed that 2 messages are sufficient in several cases. 
\end{abstract}

\section{Introduction}
The Interference Network (IN) was introduced by Carleial \cite{carleial} as a multi-terminal communication problem involving $M$ transmitters and $N$ receivers with each transmitter having independent messages for each of the $2^{N}-1$ possible non-empty subsets of the receivers. Thus, a total of $M(2^N-1)$ messages are transmitted across the channel leading to a $M(2^N-1)$ dimensional capacity region. The multiple access channel (MAC), broadcast channel (BC), interference channel (IC), and $X$ channel are all special cases of the Interference network (IN). For example, when $M=N$ and transmitter $k$ is interested in communication with only receiver $k$, we have the $M$ user IC. In the two user $X$ channel, each transmitter $Tx_i, i \in \{1,2\}$ has 2 \textit{independent} messages corresponding to the two receivers, i.e., four messages in total. 

The IC has been studied extensively in \cite{carleial,annapureddy,HK,MotKha09,ChoMotGarElg08,Kra04,EtkTseWan08,Sas04}. While the capacity region is unknown, several inner and outer bounds have been derived for the capacity region and the sum capacity \cite{HK,ChoMotGarElg08,Kra04,EtkTseWan08,Sas04}. Under some channel conditions (or interference conditions), capacity or sum capacity has been determined \cite{carleial,annapureddy,HK,MotKha09}. The $X$ channel has been studied in \cite{chiachi,HuaCadJaf09,JafSha08,MadMotKha08,CadJaf09,Hesham1} to obtain capacity region bounds and generalized degrees of freedom. 

Using all $M(2^N-1)$ messages has been observed to be important when interference networks arise as states in a half-duplex relay network \cite{MutBhaTha11,MutBhaTha}. In half-duplex relay networks, the set of transmitters and receivers at any given time instant form an interference network. The choice of rates for the  $M(2^N-1)$ messages depends on the overall information flow constraints. Therefore, a characterization of the $M(2^N-1)$ dimensional rate region is useful in flow optimization. The messages that result in optimal flow will depend on the connectivity and the channel conditions of the links. In the context of $X$ channels, it has been seen that using 2 messages is sum rate optimal under a subset of low and strong interference conditions \cite{chiachi,chiachi2}. We consider the more general IN and determine which of the 6 messages are useful for all interference conditions.  

The achievable rate region obtained using Han-Kobayashi type public-private message splitting of the 4 messages on the $X$ channel in \cite{Hesham1} provides an achievable rate region for the 2 $\times$ 2 IN. In this paper, we obtain the following results: (1) Starting from an achievable rate region in  \cite{Hesham1}, we first obtain an achievable sum rate of the 2 $\times$ 2 IN. (2) For the Gaussian interference network, we determine which of the 6 messages are sufficient for maximizing the sum rate within this rate region for the low, mixed, and strong interference conditions. It is observed that 2 messages are sufficient in several cases.

\section{Two user Discrete memoryless IN (DMIN)}
The 2 $\times$ 2 DMIN shown in Fig. \ref{2x2} is a communication model where there are 3 messages from each transmitter. The messages from $Tx_1$ are:
\begin{enumerate}
\item Direct private message $U_1$ to $Rx_1$.
\item Common message $V_1$ to both receivers $\{Rx_1,Rx_2\}$.
\item Cross private message $W_1$ to $Rx_2$.
\end{enumerate}
\begin{figure}[ht]
\begin{center}
\resizebox{2.3in}{!}{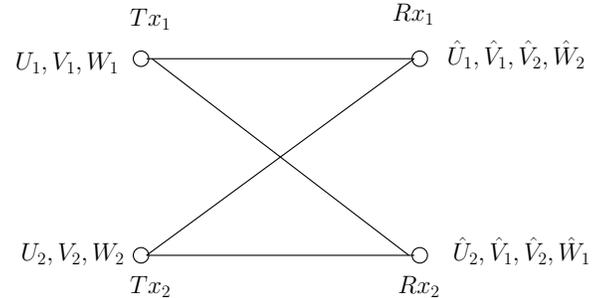}
\end{center}
\caption{2 $\times$ 2 Interference Network}
\label{2x2}
\end{figure}
Similarly the messages $U_2, V_2$ and $W_2$ originate from $Tx_2$ communicating with $Rx_2$, $\{Rx_1, Rx_2\}$ and $Rx_1$ respectively. The receiver $Rx_1$ will decode 4 messages namely $U_1,V_1,V_2$ and $W_2$. Similarly, $Rx_2$ will decode $U_2,V_1,V_2$ and $W_1$.

Although an achievable rate region for the two user DMIN has not been
explicitly reported, we can see that Han-Kobayashi (HK) \cite{HK} type
message splitting on the $X$ channel, given in \cite{Hesham1},
addresses the same problem as the IN. The HK \cite{HK} scheme,
originally proposed for two user IC, allows \textit{partial} decoding
of \textit{interference} at the unintended receiver so that a
\textit{common} part of the interference can be decoded (and
subtracted) leading to better reception of its intended signal. The
intended receiver decodes a private message, which cannot be decoded
at the other receiver, and also decodes \textit{this} common message
combining them to form its \textit{total} message. In \cite{Hesham1}, HK message splitting is applied to each of the 4 messages of the $X$ channel leading to 8 ($4\times 2$) messages and an achievable region is given. It is
easy to see that 2 \textit{public} messages originating from each
transmitter can be clubbed together as a single \textit{public}
message, resulting in a total of 6 messages. Here, we present an achievable rate region below for the 6 IN messages.\newline Let $Z=QU_1V_1W_1X_1U_2V_2W_2X_2Y_1Y_2 \in \Omega$,where $\Omega$ is the
set of all probability distributions over the
variables. $U_1,V_1,W_1,U_2,V_2,W_2$ are auxiliary random variables
and $X_1,X_2,Y_1,Y_2$ are random variables on
${\cal{X}}_1,{\cal{X}}_2,{\cal{Y}}_1,{\cal{Y}}_2$ respectively
satisfying:
\begin{enumerate}
\item $U_1, V_1, W_1, U_2, V_2, W_2$ are mutually independent given $Q$,the time sharing random variable.
\item $X_1 = f_1(U_1, V_1, W_1|Q)$, $X_2 = f_2(U_2, V_2, W_2|Q)$, where $f_1$ and $f_2$ are deterministic functions of their arguments.
\end{enumerate}

Let $\mathcal{R}(Z)$ denote the rate region formed by the six tuple rate $(R_{U_1}, R_{V_1}, R_{W_1}, R_{U_2}, R_{V_2}, R_{W_2})$ satisfying the following constraints:
\begin{equation}
R_{S_1} = \sum_{s \in S_1} R_s \leq I(S_1;Y_1|\bar{S}_1,Q) ~~~~ \forall S_1, \label{IN1}\\
\end{equation}
where $S_1$ is any non-empty subset of $M_1 = \{ U_1, V_1, V_2, W_2\}$, and $\bar{S}_1 = M_1 \backslash S_1$. Since there are 15 possible subsets $S_1$, we have 15 constraints. Similarly, considering $Rx_2$, we get another 15 constraints corresponding to each non-empty subset $S_2$ of $M_2 = \{ U_2, V_1, V_2, W_1\}$. For example, one of 30 constraints is $R_{U_2} + R_{V_2} + R_{W_1} \leq I(U_2, V_2, W_1;Y_2|V_1, Q)$.

Let $\mathcal{R}_{IN}$ be the closure of $\bigcup_{Z\in \Omega}\mathcal{R}(Z)$. Then, any rate tuple in $\mathcal{R}_{IN}$ is achievable for the two user DMIN. The proof of achievability uses jointly typical decoding and is similar to the proof in \cite{HK}.

\section{Achievable sum rate}
Let the sum rate $S=R_{U_1} + R_{V_1} + R_{W_1} + R_{U_2} + R_{V_2} + R_{W_2}$.

\begin{theorem} The achievable sum rate $S$ is bounded as follows. 
\begin{equation}
S  \leq \min\{ T_1, T_2, T_3, T_4\},
\label{sumratebounds}
\end{equation}
where
\begin{eqnarray}
\nonumber
T_1 &=&I(U_2,V_1,V_2,W_1;Y_2|Q)+I(U_1,W_2;Y_1|V_1,V_2,Q),\nonumber\\
T_2 &=&I(U_1,V_1,V_2,W_2;Y_1|Q)+I(U_2,W_1;Y_2|V_1,V_2,Q),\nonumber\\
T_3 &=&I(U_1,V_2,W_2;Y_1|V_1,Q)+I(U_2,V_1,W_1;Y_2|V_2,Q).\nonumber\\
T_4 &=&I(U_1,V_1,W_2;Y_1|V_2,Q)+I(U_2,V_2,W_1;Y_2|V_1,Q),\nonumber
\end{eqnarray}
\end{theorem}
\begin{proof}
See Appendix \ref{claim1proof}.
\end{proof}

It is worth noting that the sum rate can be bounded by several
expressions using the constraints in $(\ref{IN1})$. For example, by
adding the bounds on $R_{U_1} + R_{V_1} + R_{V_2}$, $R_{U_2} + R_{W_1}$ and $R_{W_2}$ in three of the constraints, we can get a bound
for $S$. One can obtain similar bounds by several groupings of the 6 rate components of sum rate and adding the corresponding constraints from the rate region.  The sum rate is bounded by the minimum of all such bounds. In the proof, it is shown that show that only 4 of the
combinations are useful and the others are redundant.

\section{Gaussian Interference Network (GIN)}
The standard form for the Gaussian IN \cite{carleial} is  
\begin{eqnarray}
\nonumber
Y_1&=&X_1+h_2X_2+Z_1\\
Y_2&=&X_2+h_1X_1+Z_2
\label{gin}
\end{eqnarray}
where $Z_1,Z_2 \sim N(0,1)$. Power constraint $P_1,P_2$ are imposed on $Tx_1$ and $Tx_2$ respectively. The channel (or interference) conditions for the two user GIN can be classified into the following cases.
\begin{enumerate}
\item \textit{Low Interference (LI)}: $0\leq h_1\leq 1,0\leq h_2\leq 1$.
\item \textit{Mixed Interference (MI)}: $0\leq h_1\leq 1,h_2\geq 1$ or $0\leq h_2\leq 1,h_1\geq 1$.
\item \textit{Strong Interference (SI)}: $1\leq h_1^2\leq P_2+1,1\leq h_2^2 \leq P_1+1$.
\item \textit{Very Strong Interference (VSI)}: $h_1^2\geq P_2+1,h_2^2\geq P_1+1$.
\end{enumerate}
For the GIN, the DMIN rate region can be extended as follows. We consider the non-time sharing case, where $Q=\phi$, a constant. Further, we limit ourselves to $X_1=U_1+V_1+W_1,X_2=U_2+V_2+W_2$, where $U_1,V_1,W_1,U_2,V_2,W_2$ are independent Gaussian codebooks. Let us denote this rate region as $\mathcal{R}_{GIN}$. We employ superposition coding \cite{cover_BC} at the transmitters with power distribution defined as follows. Messages $U_i$, $V_i$, and $W_i$ are transmitted using powers  $\alpha_iP_i$, $\beta_iP_i$, and $\gamma_iP_i$, respectively, for $i = 1, 2$. Also, $\alpha_i+\beta_i+\gamma_i=1$. Let $I_1=1+h_2^2\alpha_2P_2+\gamma_1P_1$ and
$I_2=1+h_1^2\alpha_1P_1+\gamma_2P_2$. Let $C(x)=0.5\log_2(1+x)$. An
achievable rate region for the two user GIN is once again defined by
30 constraints. The 15 constraints corresponding equation
(\ref{IN1}) are given by:
\begin{eqnarray*}
\fontsize{8pt}{11pt}
R_{S_1} = \sum_{s \in S_1} R_s &\leq&C\left(\frac{\sum_s P_s}{I_1}\right),\nonumber\\
\end{eqnarray*}
where $P_{U_1} = \alpha_1 P_1$, $P_{V_1} = \beta_1 P_1$, $P_{V_2} = h_2^2 \beta_2 P_2$, and $P_{W_2} = h_2^2 \gamma_2 P_2$. Similarly, 15 more constraints can be written considering the rate constraints for $Rx_2$. Note that
\begin{center}
$\mathcal{R}_{GIN}\subseteq\mathcal{R}_{GIN_Q}\subseteq\mathcal{R}_{IN} \subseteq C_{IN}$,
\end{center}
where $\mathcal{R}_{GIN_Q}$ is the rate-region with optimal time sharing ($Q\neq\phi$) strategy and Gaussian input. $\mathcal{R}_{IN}$ is the optimal time sharing strategy with optimal input distribution and $C_{IN}$ is the capacity of the IN.

\section{Achievable Sum Rates in GIN}

In this section, we determine which of the 6 messages in the IN are
useful in maximizing the sum rate for the various interference
conditions. In \cite{chiachi}, a similar question was answered for the
$X$ channel in a subset of the low and strong interference regimes
extending the result for IC in \cite{annapureddy}. Since the sum
capacity of an IN is unknown, we first study the maximum sum rate
within the achievable rate region described in the previous
section as summarized in Table
\ref{results-table}. In Section \ref{sum-cap}, we show that the sum
capacity is indeed achieved for some mixed interference conditions.
\begin{table}[h]
\caption{Summary of Results}
\fontsize{9pt}{11pt}
\label{tabl}
\begin{tabular}{|l|c|l|}
\hline
\bfseries Region & \bfseries Sub-region & \bfseries Message-set\\
\hline
L.I &-& $U_1,V_1,U_2,V_2$\\\cline{2-3}
&$|h_1(1+h_2^2P_2)+h_2(1+h_1^2P_1)|\leq 1$ & $U_1,U_2$\\
\hline
M.I &$0\leq h_1\leq 1,h_2\geq 1$& $U_1,W_2$ \\\cline{2-3}
&$0\leq h_2\leq 1,h_1\geq 1$& $U_2,W_1$\\
\hline
S.I &-& $W_1,V_1,W_2,V_2$\\
\hline
V.S.I &-& $W_1,V_1,W_2,V_2$\\\cline{2-3}
&$|h_1^{-1}(1+P_2)+h_2^{-1}(1+P_1)|\leq 1$ & $W_1,W_2$\\
\hline
\end{tabular}
\label{results-table}
\end{table}

For the GIN, the terms $T_1$, $T_2$, $T_3$, and $T_4$ in the sum rate bound in equation (\ref{sumratebounds}) are:
\begin{eqnarray}
\nonumber
T_1=C\left(\frac{\alpha_1P_1+h_2^2\gamma_2P_2}{1+h_2^2\alpha_2P_2+\gamma_1P_1}\right) + C\left(\frac{\bar{\gamma}_2P_2+h_1^2\bar{\alpha}_1P_1}{1+h_1^2\alpha_1P_1+\gamma_2P_2}\right),&&\nonumber\\
T_2=C\left(\frac{\bar{\gamma}_1P_1+h_2^2\bar{\alpha}_2P_2}{1+h_2^2\alpha_2P_2+\gamma_1P_1}\right) + C\left(\frac{\alpha_2P_2+h_1^2\gamma_1P_1}{1+h_1^2\alpha_1P_1+\gamma_2P_2}\right),&&\nonumber\\
T_3=C\left(\frac{\alpha_1P_1+h_2^2\bar{\alpha}_2P_2}{1+h_2^2\alpha_2P_2+\gamma_1P_1}\right) + C\left(\frac{\alpha_2P_2+h_1^2\bar{\alpha}_1P_1}{1+h_1^2\alpha_1P_1+\gamma_2P_2}\right),&&\nonumber\\
T_4=C\left(\frac{\bar{\gamma}_1P_1+h_2^2\gamma_2P_2}{1+h_2^2\alpha_2P_2+\gamma_1P_1}\right) + C\left(\frac{\bar{\gamma}_2P_2+h_1^2\gamma_1P_1}{1+h_1^2\alpha_1P_1+\gamma_2P_2}\right),&&\nonumber
\end{eqnarray}
where $\bar{\alpha}_i = 1 - \alpha_i$ and $\bar{\gamma}_i = 1 - \gamma_i$.

\subsection{Mixed Interference}
There are two cases for Mixed Interference: (i) $0 \leq h_2 \leq 1,h_1
\geq 1$, and (ii) $0 \leq h_1 \leq 1,h_2 \geq 1$. 

\begin{theorem}
\begin{enumerate}
\item For case (i), the achievable sum rate is maximized by transmitting only $U_2$ and $W_1$, both to $Rx_2$. The sum rate achieved is the MAC sum capacity at $Rx_2$ $=C(h_1^2P_1 + P_2)$. 
\item For case (ii), the achievable sum rate is maximized by transmitting only $U_1$ and $W_2$, both to $Rx_1$. The sum rate achieved is the MAC sum capacity at $Rx_1$ $=C(h_2^2P_2 + P_1)$. 
\end{enumerate}
\end{theorem}
\begin{proof}
The proof of statement (1) is in Appendix \ref{mixed-proof}. The other statement can be proved similarly by swapping indices 1 and 2.
\end{proof}

\subsection{Low Interference}
\begin{theorem}
\begin{enumerate}
\item Let $T_i = t_i$ for $i = 1, 2, 3, 4$ when the power sharing fractions are $\alpha_1$, $\beta_1$, $\gamma_1$, $\alpha_2$, $\beta_2$, $\gamma_2$. Let $T_i = t_i^{'}$ when the power sharing fractions are $\alpha_1^{'} = \alpha_1$, $\beta_1^{'} = \beta_1 + \gamma_1$, $\gamma_1^{'} = 0$, $\alpha_2^{'} = \alpha_2$, $\beta_2^{'} = \beta_2 + \gamma_2$, $\gamma_2^{'} = 0$. Then $t_i^{'} \ge t_i$ for $i = 1, 2, 3, 4$ if $0 \leq h_1, h_2 \leq 1$. 
\item Messages $W_1$ and $W_2$ are not required to maximize the sum rate when $0 \leq h_1, h_2 \leq 1$.
\end{enumerate}
\end{theorem}
\begin{proof}
See Appendix \ref{low-proof}.
\end{proof}

From the theorem above, it is clear that only 4 messages $U_1$, $U_2$, $V_1$, and $V_2$ are required (as in the IC) to maximize sum rate in the low interference regime (as mentioned in Table \ref{results-table}). Further, in \cite{annapureddy}, it is also proved that in the IC, for channel conditions satisfying
\begin{equation}
|h_1(1+h_2^2P_2)+h_2(1+h_1^2P_1)|\leq 1,
\label{lowint}
\end{equation}
encoding messages $U_1, U_2$ alone using Gaussian codebooks and treating interference as noise at each receiver is sum-capacity optimal. In \cite{chiachi}, this result is extended to the $X$ channel as well. Having shown that $\gamma_i=0, i\in\{1,2\}$, the same result also holds for $\mathcal{R}_{GIN}$.

\subsection{Strong Interference}
The conditions for strong interference are $1\leq h_1^2\leq P_2+1,
1\leq h_2^2\leq P_1+1$. Define $X_1'=h_1X_1, X_2'=h_2X_2$. Now the
equation (\ref{gin}) can be rewritten as
\begin{eqnarray}
\nonumber
Y_1&=&\frac{X_1'}{h_1}+X_2'+Z_1\\
Y_2&=&\frac{X_2'}{h_2}+X_1'+Z_2
\label{modgin}
\end{eqnarray}
In the strong interference regime, $\frac{1}{h_1}\leq 1,\frac{1}{h_2}\leq 1$.Thus, we now have an equivalent GIN in \textit{low interference} corresponding to the each strong interference GIN. $X_2^{'}$ now carries the \textit{direct} messages to $Rx_1$ and $X_1^{'}$ carries the \textit{cross private} message to $Rx_1$. The roles of $U_i$ and $W_i$, $i\in\{1,2\}$ interchange from their respective roles in the \textit{low interference} regime. Therefore, $\alpha_i = 0$ (i.e., $U_1$ and $U_2$ are not necessary) for maximizing the sum rate.

\subsection{Very Strong Interference}
\noindent In this case, we can make the following observations:\\
\noindent (1) The conclusions for strong interference that $\alpha_i=0$ holds here as well.\\

\noindent (2) For the model (\ref{modgin}) given above, let $P_1' =$ var$(X_1') = h_1^2P_1$, $P_2' =$ var$(X_2') = h_2^2P_2$, $h_1^{'} = 1/h_1$, and $h_2^{'} = 1/h_2$. We already know that, if
\begin{equation}
|h_1^{'}(1+h_2^{'2}P_2^{'})+h_2^{'}(1+h_1^{'2}P_1^{'})|\leq 1
\end{equation}
then this corresponds the sub-region in low interference discussed earlier. In this region, only messages $W_1$ and $W_2$ are sufficient to maximize the sum rate. The condition can be rewritten in terms of the original channel and power variables are
\begin{equation}
\left|\frac{1+P_2}{h_1}+\frac{1+P_1}{h_2}\right|\leq 1.
\end{equation}
Thus, there is a sub-region within the very strong interference
satisfying the above condition where only the 2 messages $W_1$ and
$W_2$ are necessary to maximize sum rate in $\mathcal{R}_{GIN}$.

\section{Conclusions}
Using an achievable rate region similar to the Han-Kobayashi region,
we obtain an achievable sum rate for a 2 $\times$ 2
GIN. We determine that at most 4 (out of 6) messages are sufficient
for maximizing the sum rate within this rate region for all channel
conditions. Also, in no case is more than one private message
transmitted from any transmitter. It is also observed that 2 messages
are sufficient in several cases -- mixed interference, and sub-regions
of low and very strong interference regions. 

\appendices
\section{Proof of Theorem 1}
\label{claim1proof}

We know $S=R_{U_1} + R_{V_1} + R_{W_1} + R_{U_2} + R_{V_2} + R_{W_2}$
is the sum of 6 different rates. The rate region constraints are
constraints on the sum of 1 or 2 or 3 or 4 of these rates. The 15
constraints at each receiver comprise of 4 single rate constraints, 6
on sum of 2 rates, 4 on sum of 3 rates and one on the sum of 4
rates. In order to obtain a bound on $S$, we can choose 2 or more
constraints from the 30 available constraints appropriately.

First, we observe that only one constraint needs to be chosen from
each group of 15 constraints (i.e. for each receiver). This is
because:
\begin{itemize}
\item The messages are independent given $Q$ by assumption.
\item If more than one constraint is chosen from the same group
  (corresponding to the same receiver), a single tighter constraint
  can be obtained in the following manner. If 2 constraints are chosen
  from equation (\ref{IN1}) corresponding to 2 disjoint subsets $C_1$
  and $C_2$ of $M_1$, we get the sum constraint
  $I(C_1;Y_1|\bar{C}_1,Q) + I(C_2;Y_1|\bar{C}_2,Q)$. However, $I(C_1
  \bigcup C_2; Y_1 | \overline{C_1 \bigcup C_2}, Q)$ is a tighter bound
  (due to the independent messages assumption) and is also one of the 15
  constraints.
\end{itemize}

Now, there are only 4 possible combinations of 2 constraints with one
from each group of 15 constraints. These are the 4 stated bounds
$T_1$, $T_2$, $T_3$, and $T_4$ in the theorem. A similar approach is
also used in \cite{ChoMotGarElg08} for reducing the number of sum
constraints in an interference channel setting.

\section{Proof of Theorem 2}
\label{mixed-proof}
In order to prove statement (1), it is sufficient to show that any one
of the $T_i$'s is less than or equal to $C(h_1^2P_1 + P_2)$. This is
because (i) any bound on a $T_i$ is also a bound on $S$, and (ii) we
know that $C(h_1^2P_1 + P_2)$ can be achieved using messages
$U_2$ and $W_1$ alone.

We can show that $T_1 \leq C(h_1^2P_1 + P_2)$ for any $\alpha_i, \beta_i, \gamma_i$ and $0 \leq h_2 \leq 1$, $h_1 \ge 1$. Proving $T_1 \leq C(h_1^2P_1 + P_2)$ can be shown (using the monotonicity of the $\log$ function) to be equivalent to showing 
\[
\frac{1 + \alpha_1 P_1 + \gamma_1P_1 + h_2^2 \alpha_2 P_2 + h_2^2 \gamma_2 P_2}{(1 + \gamma_1P_1 + h_2^2 \alpha_2 P_2)(1 + \gamma_2P_2 + h_1^2 \alpha_1 P_1)} \leq 1.
\]
This is the same as showing
\[
 \alpha_1 P_1 + h_2^2 \gamma_2 P_2 \leq (h_1^2 \alpha_1 P_1 + \gamma_2P_2)(1 + \gamma_1P_1 + h_2^2 \alpha_2 P_2).
\]
This condition is true for $0 \leq h_2 \leq 1$, $h_1 \ge 1$.

\section{Proof of Theorem 3}
\label{low-proof}
\noindent {\em Comparison of $t_1$ and $t_1^{'}$:} 
\[
t_1^{'} = C\left(\frac{\alpha_1P_1}{1+h_2^2\alpha_2P_2}\right) + C\left(\frac{P_2+h_1^2\bar{\alpha}_1P_1}{1+h_1^2\alpha_1P_1}\right)
\]
Proving $t_1 \leq t_1^{'}$ can be shown (using the monotonicity of the $\log$ function) to be equivalent to showing 
\[
\frac{A_1 + \gamma_1P_1 + h_2^2\gamma_2 P_2}{(A_2 + \gamma_1 P_1)(A_3 + \gamma_2 P_2)} \leq \frac{A_1}{A_2 A_3}, 
\]
where $A_1 = 1 + \alpha_1 P_1 + h_2^2 \alpha_2 P_2$, $A_2 = 1 + h_2^2 \alpha_2 P_2$, and $A_3 = 1 + h_1^2 \alpha_1 P_1$. Equivalently, we need to show
\[
A_2A_3(\gamma_1P_1 + h_2^2\gamma_2P_2) \leq \gamma_1P_1A_1A_3 + \gamma_2P_2 A_1A_2 + \gamma_1P_1\gamma_2P_2A_1.
\]
This is shown by comparing the first 2 terms using: (a) $A_1 \ge A_2$, (b) $A_1 \ge A_3$ when $0 \leq h_1 \leq 1$, and (c) $0 \leq h_2 \leq 1$.

\noindent {\em Comparison of $t_2$ and $t_2^{'}$:} This is similar to the comparison of $t_1$ and $t_1^{'}$ expect that the indices 1 and 2 are interchanged in the expressions for $t_2$ and $t_2^{'}$ when compared with $t_1$ and $t_1^{'}$.

\noindent {\em Comparison of $t_3$ and $t_3^{'}$:} 
\[
t_3^{'} = C\left(\frac{\alpha_1P_1+h_2^2\bar{\alpha}_2P_2}{1+h_2^2\alpha_2P_2}\right) + C\left(\frac{\alpha_2P_2+h_1^2\bar{\alpha}_1P_1}{1+h_1^2\alpha_1P_1}\right)
\]
Clearly, $t_3$ is always less than or equal to $t_3^{'}$ since only denominator is reduced (by setting $\gamma_1 = \gamma_2 = 0$) in both the arguments for $C(.)$ in $t_3^{'}$.

\noindent {\em Comparison of $t_4$ and $t_4^{'}$:} 
\[
t_4^{'} = C\left(\frac{P_1}{1+h_2^2\alpha_2P_2}\right) + C\left(\frac{P_2}{1+h_1^2\alpha_1P_1}\right).
\]
Proving $t_4 \leq t_4^{'}$ can be shown (using the monotonicity of the $\log$ function) to be equivalent to showing 
\[
\left(\frac{A_1 + h_2^2 \gamma_2P_2}{A_2 + \gamma_2 P_2}\right)\left( \frac{A_3 + h_1^2 \gamma_1P_1}{A_4 + \gamma_1 P_1}\right) \leq \frac{A_1}{A_2} . \frac{A_3}{A_4},
\]
where $A_1 = 1+ P_1 + h_2^2 \alpha_2P_2$, $A_2 = 1 + h_2^2 \alpha_2P_2$, $A_3 = 1+ P_2 + h_1^2 \alpha_1P_1$, and $A_4 = 1 + h_1^2 \alpha_1P_1$. This is true for $0 \leq h_1, h_2 \leq 1$.

\bibliographystyle{IEEEtran}
\bibliography{icc2012murali}
\end{document}